\documentclass[11pt]{article} 

\usepackage[utf8]{inputenc} 
\usepackage[toc, page]{appendix}
\usepackage{amsfonts}
\usepackage{bm}
\usepackage{amsmath}
\usepackage{bbm}
\usepackage{eurosym}
\numberwithin{equation}{section}
\numberwithin{figure}{section}
\numberwithin{table}{section}

\usepackage[ruled,vlined, linesnumbered]{algorithm2e}

\usepackage[total={6.5in, 8in}]{geometry} 
\usepackage{graphicx} 
\usepackage{float}
\usepackage{caption}
\usepackage{subcaption}
\usepackage{rotating}

\usepackage{booktabs} 
\usepackage{array} 
\usepackage{paralist} 
\usepackage{verbatim} 
\usepackage{hyperref}
\hypersetup{pdftex,colorlinks=true,allcolors=blue}
\usepackage{hypcap}
\usepackage{amsthm}

\usepackage{makecell}

\usepackage{fancyhdr} 
\usepackage{sectsty}
\allsectionsfont{\sffamily\mdseries\upshape} 

\usepackage[nottoc,notlof,notlot]{tocbibind} 
\usepackage[titles,subfigure]{tocloft} 

\usepackage[font=small,labelfont=bf,tableposition=top]{caption}

\usepackage{authblk}
\usepackage[export]{adjustbox}
\usepackage{mathtools}

\newtheorem{definition}{Definition}[section]
\newtheorem{remark}[definition]{Remark}

\newtheorem{approximation}{Approximation}[section]
\newtheorem{theorem}{Theorem}[section]

\newcommand{\PreserveBackslash}[1]{\let\temp=\\#1\let\\=\temp}
\newcolumntype{C}[1]{>{\PreserveBackslash\centering}p{#1}}
\newcolumntype{R}[1]{>{\PreserveBackslash\raggedleft}p{#1}}
\newcolumntype{L}[1]{>{\PreserveBackslash\raggedright}p{#1}}

\DeclareCaptionLabelFormat{andtable}{#1~#2  \&  \tablename~\thetable}

\providecommand{\keywords}[1]
{
  \small	 
  \textbf{\textit{Keywords---}} #1
}

\title{Multi-asset market making under the quadratic rough Heston \thanks{
  This work benefits from the financial support of the Chaires Machine Learning \& Systematic Methods. The authors would like to thank Marouane Anane and Alexandre Davroux for very useful comments.
}}

\author{Mathieu Rosenbaum $^1$ \\ \vspace{-1em} mathieu.rosenbaum@polytechnique.edu
\and
Jianfei Zhang $^{1,2}$ \\ \vspace{-1em} jianfei.zhang@polytechnique.edu}
\date{%
$^1$ \footnotesize École polytechnique, CMAP, Institut Polytechnique de Paris, 91120 Palaiseau, France \\%
$^2$ \footnotesize Exoduspoint Capital Management, 32 Boulevard Haussmann, 75009 Paris, France \\[2ex]%
\today}
\begin{document}
\maketitle

\begin{abstract}
  Given the promising results on joint modeling of SPX/VIX smiles of the recently introduced quadratic rough Heston model, 
  we consider a multi-asset market making problem on SPX and its derivatives, \textit{e.g.} VIX futures, SPX and VIX options. 
  The market maker tries to maximize its profit from spread capturing while controlling the portfolio's inventory risk, which can be fully explained 
  by the value change of SPX under the particular setting of the quadratic rough Heston model. The high dimensionality of the 
  resulting optimization problem is relaxed by several approximations. An asymptotic closed-form solution can be obtained.
  The accuracy and relevance of the approximations are illustrated through numerical experiments.
\end{abstract}
\keywords{Multi-asset market making, quadratic rough Heston, quadratic approximation}

\section{Introduction}
\label{sec:intro}
It is well known that the constant volatility assumption in the celebrated Black-Scholes model is not consistent with empirical observations of 
financial time series and uneven implied volatility surfaces. 
Tackling this issue by regarding the volatility as a continuous-time random process, stochastic volatility models are able to reproduce several 
stylized facts of historical data, 
such as the fat-tailed distribution of returns and the volatility clustering phenomenon. However, the shapes of the implied volatility surfaces 
generated by conventional stochastic volatility models, such as the Hull and White, Heston, and SABR models, usually differ substantially from those of empirical observations, see for instance \cite{bayer2016pricing}. 
The rough volatility paradigm brings new solutions, being able to achieve superior fits of implied volatility surfaces than the former models \cite{bayer2016pricing,el2019roughening, gatheral2018volatility}. 
The recently introduced quadratic rough Heston (QRH) model shows its potential in calibrating jointly SPX and VIX smiles \cite{gatheral2020quadratic}. It models the price of an asset $S$ (here the SPX)
and its spot variance $V$ under risk-neutral measure as
\begin{equation*}
  dS_t = S_t\sqrt{V_t}dW_t, \quad V_t=a(Z_t - b)^2 + c,  
\end{equation*}
where $W$ is a Brownian motion, $a, b, c$ are all positive constants and $Z_t$ is defined as
\begin{equation}
  Z_t = \int_0^t \lambda\frac{(t-s)^{\alpha - 1}}{\Gamma(\alpha)}\big(\theta_0(s) - Z_s\big)ds + 
  \int_0^t\eta\frac{(t-s)^{\alpha-1}}{\Gamma(\alpha)}\sqrt{V_s}dW_s  \, ,
\label{eq:Z_t}
\end{equation} 
where $\alpha\in(1/2, 1), \lambda>0$, $\eta>0$, and $\theta_0(\cdot)$ is a deterministic function. The fractional kernel $K(t) = \frac{t^{\alpha-1}}{\Gamma(\alpha)}$ enables us to generate rough volatility dynamics. It is also used in the rough Heston model, under which the volatility trajectories have almost surely H$\ddot{\text{o}}$lder regularity $\alpha - 1/2 -\varepsilon$, for any $\varepsilon > 0$ \cite{el2019characteristic}. This actually recalls the observation in \cite{gatheral2018volatility} that the dynamic of log-volatility is similar to that of a fractional Brownian motion with Hurst parameter of order 0.1. Importantly, the QRH model gives a natural way to encode the 
\textit{strong Zumbach effect} \cite{dandapani2021quadratic}, which means that the conditional law of future volatility depends on the past not
only through past volatility trajectories but also through past returns.

\vskip 0.15in
\noindent
Note that in this model, only one Brownian motion is involved to generate the stochastic nature of the model and the dynamic of the volatility 
can be fully explained by past returns. While this is quite different from conventional settings, where volatility is exposed to 
additional random factors, satisfactory results on market data are reported in \cite{gatheral2020quadratic,rosenbaum2021deep}.
It is also emphasized in \cite{guyon2022volatility} that price returns largely explain volatility. An extensive empirical study on the formation
process of realized volatility is presented in \cite{rosenbaum2022universality}, where the predictive power of price returns on future realized volatility 
is underlined. 
The QRH model gives opportunities to model consistently SPX derivatives, including VIX futures, SPX and VIX options. 
This motivates us to consider a multi-asset market making problem on SPX\footnote{SPX, or formally S\&P 500 index, is not tradable. When no ambiguity occurs, in this paper it means the liquid index-tracking products such as index funds, ETF, and futures, 
which are highly correlated with each other.} and the aforementioned derivatives of SPX. 
Today these assets are actively traded and most of the volume comes from electronic trading platforms based on limit order books, 
which justifies the relevance of developing automated market making algorithms.

\vskip 0.15in
\noindent
The academic literature on optimal market making dates back to the 80s. In \cite{ho1981optimal}, the challenge faced by a market maker, 
which consists in maximizing its profit from spread capturing with a particular focus on execution uncertainty and inventory risk, is formulated as a dynamic 
programming problem. Approximate solutions are given for specific functional forms of the order arrival rates. The approach is revived in the seminal 
work \cite{avellaneda2008high} with tools of stochastic optimal control. Since then, an extensive literature on optimal market
making has been developed. It is shown in \cite{gueant2013dealing} that under inventory constraints, the Hamilton-Jacobi-Bellman (HJB) equations
associated with the control problem of \cite{avellaneda2008high} can be reduced to a linear system of ordinary differential equations (ODEs), for which
the asymptotic of the closed-form solution are given when the considered time horizon tends to infinity. Various features, such as model misspecification, 
short-term alpha signals, and microstructural characteristics, 
are investigated in \cite{cartea2017algorithmic,cartea2014buy,cartea2018algorithmic}. In these papers, a mean-variance type objective function is considered instead of the Von Neumann-Morgenstern expected utility as in \cite{avellaneda2008high,gueant2013dealing}. In \cite{gueant2017optimal}, 
the results in \cite{gueant2013dealing} are extended to more general order execution rates and the two aforementioned objective functions are 
reconciled. It is shown that the HJB equations associated with the two classes of objective functions can be sorted out into the same 
family of ODEs. 

\vskip 0.15in
\noindent
Multi-asset market making problems are usually exposed to the curse of dimensionality. The conventional grid methods typically used in the case of single-asset become 
inadmissible given the huge amount of computational resources required. In \cite{bergault2021size}, a method based on factor decomposition is developed to reduce the dimension
of the problem. Deep neural networks are considered in \cite{gueant2019deep} to approximate the value function and the optimal quoting strategy. Closed-form approximations
are obtained in \cite{bergault2021closed} replacing the Hamiltonian functions with quadratic ones. A similar idea is followed in \cite{baldacci2020approximate} in the context of 
option market making. We refer to \cite{cartea2015algorithmic, gueant2016financial} for a panorama on algorithmic market making. 

\vskip 0.15in
\noindent
All the models proposed in the above papers consider continuous control variables, the quotes' distances with respect to the reference price, 
which makes them more adapted to quote-driven markets such as corporate bond markets and FX markets. The case of the markets ruled by limit order books is addressed in \cite{guilbaud2013optimal,guilbaud2015optimal}, where the market maker decides whether to submit limit orders at the best limits or to get immediate execution using 
market orders. More dedicated microstructure models reflecting some 
important empirical observations at the microstructural scale are chosen in \cite{fodra2015high,lu2018order}.

\vskip 0.15in
\noindent
In this paper, we consider an agent continuously quoting for SPX and its derivatives including VIX futures, SPX and VIX options, 
via limit orders on both the bid and ask sides. The market maker tries to maximize its expected gain from spread capturing while controlling the inventory risk.
In the option market making problems considered in \cite{baldacci2021algorithmic, baldacci2020approximate}, 
the underlying follows a one-factor stochastic volatility model and 
perfect Delta hedging is assumed for European options. The market maker then penalizes the option portfolio's total Vega. In the QRH model, the price dynamics of all the aforementioned derivatives can be explained by the value variation of SPX. Thus, in the context of market making, the principal
inventory risk is summarized by the sensitivities of their price change with respect to that of SPX.  Even if theoretically this risk can be hedged out
perfectly using only SPX as shown in \cite{rosenbaum2021deep}, continuous hedging is not realistic and the cost of frequent 
hedging is far from negligible. More particularly, as the typical inventory process of market making shows a mean reverting behavior around zero,
frequent hedging results in unnecessary buys and sells of SPX. In our approach, instead of hedging out completely the net inventory risk of the held derivatives 
with SPX, the agent also market makes on SPX with limit orders and some net inventory risk is authorized.

\vskip 0.15in
\noindent
The market making basket selected here is essentially made of large tick assets, for which the effective spread is almost always equal to one tick. Thus the quoting strategy of the market maker considered in this paper consists in whether to send limit orders on the two best limits given its 
portfolio inventory, \textit{i.e.} whether to make a market on the best limits, instead of choosing the optimal quoting price as most of the models inspired from that of \cite{avellaneda2008high}.
The resulting modeling framework is in the spirit of \cite{guilbaud2013optimal}, while the dimension of the problem is much higher. 
The high dimensionality comes from two aspects. 
First, we use here the multi-factor approximation of the QRH model, which is Markovian and suitable for the computation of hedging quantities \cite{abi2019multifactor, rosenbaum2021deep}. 
However, this brings additional model states. Second, the portfolio of the market maker should be scaled to dozens of assets, including futures and options of various specifications.

\vskip 0.15in
\noindent
To relax the first one, we assume the volatility of SPX and the price sensitivities of other assets against SPX to be constant 
during the time horizon of market making. At first sight, it seems to deviate from the important feature of the QRH model that volatility spikes can happen 
in the case of significant price trends. However, this assumption is relevant given the short time horizon of market making problems, which is 
usually less than several hours. In fact, we always follow the approximated QRH model for the computation of related quantities when switching from one market making period to another, \textit{e.g.} at the beginning of each day, and we only take them as constant during the following period. As for the multi-asset nature of the problem,
we apply the idea introduced in \cite{bergault2021closed} of approximating the HJB equation with another one whose asymptotic closed-form solution can be deduced.
Having closed-form approximate solutions makes the recalibration of the algorithm very efficient, and 
thus it is even possible to design an ``online'' market making algorithm, \textit{i.e.} the algorithm is fed regularly with the latest parameter values inside each market making period.

\vskip 0.15in
\noindent
The paper is organized as follows. In Section \ref{sec:prob_desc}, we recall first the multi-factor approximation of the QRH model, and then 
describe the considered multi-asset market making problem. Several approximations are introduced in Section \ref{sec:prob_approx} to reduce the 
dimension of the problem. The asymptotic value function associated with the approximated problem is given in closed-form. Finally, in Section \ref{sec:num_res},
we evaluate the relevance of these approximations through numerical experiments.

\section{Description of the problem}
\label{sec:prob_desc}
\subsection{Multi-factor approximation of the QRH model}
\label{sec:qrh_model}
Inspired by \cite{abi2019multifactor}, in \cite{rosenbaum2021deep} a multi-factor approximation of the QRH model is introduced to make the model Markovian, and an efficient calibration procedure based on
deep learning is designed. Essentially, the multi-factor approximation consists in replacing $K(t):=\frac{t^{\alpha-1}}{\Gamma(\alpha)}$ in Equation (\ref{eq:Z_t}) with 
the approximated kernel function $K^n(t)=\sum_{i=1}^nc_ie^{-\gamma_it}$. Then we have
\begin{equation}
    dZ^i_t = (-\gamma_iZ^i_t - \lambda \sum_{i=1}^nc_iZ^i_t)dt + \eta\sqrt{V_t}dW_t, \qquad Z^i_0 = z^i_0 \, ,
  \label{eq:qrh_model}
\end{equation}
and $Z_t = \sum_{i=1}^nc_iZ^i_t$. The $2n$ parameters $(c_i, \gamma_i)_{i=1,\cdots, n}$ are explicit functions of $\alpha\in (\frac{1}{2}, 1)$, instead of being free parameters to calibrate, 
see \cite{rosenbaum2021deep} and the references therein for more details.
Now given the Markovian nature of the approximated version of the model, the price of many derivatives  of SPX at time $t$, \textit{e.g.} VIX futures, vanilla SPX and VIX options, can be obtained as a function 
of $X_t:=(S_t, Z_t^1,\cdots, Z_t^n)$. In fact, with $E_t[\cdot]:=E[\cdot |\mathcal{F}_t]$ being the conditional expectation under the risk-neutral measure, we have the following representations:
\begin{itemize}
  \item VIX future with expiration $T$: $\mathbb{E}_t[\text{VIX}_T]$ with $\text{VIX}_t:=100\sqrt{-2\mathbb{E}_t[\log(S_{t+\Delta}/S_t)]}$, $\Delta=30$ days.
  \item Vanilla SPX call or put with expiration $T$ and strike $K$: $\mathbb{E}_t[(S_T - K)_{+}]$ or $\mathbb{E}_t[(K-S_T)_{+}]$. 
  \item Vanilla VIX call or put with expiration $T$ and strike $K$: $\mathbb{E}_t[(\text{VIX}_T - K)_{+}]$ or $\mathbb{E}_t[(K-\text{VIX}_T)_{+}]$. 
\end{itemize}
\noindent
Let $P(t, X)$ denote one of these quantities. Under usual regularity conditions\footnote{In practice the price of any derivative concerned here can be approached by a regular approximator, such as the neural network used in \cite{rosenbaum2021deep}.}, we can write
\begin{equation*}
  dP(t, X) = \delta_tdS_t \, ,
\end{equation*}
where 
\begin{equation*}
  \delta_t = \frac{\partial P(t, X)}{\partial S} + \frac{\eta}{S}\sum_{i=1}^n\frac{\partial P(t, X)}{\partial Z^i} \, .
\end{equation*}
An example with vanilla SPX call is taken in \cite{rosenbaum2021deep} to show the computation of $\delta_t$ using neural networks. 
The same approach can be extended to VIX futures and other vanilla options\footnote{For SPX we have naturally $\delta_t\equiv 1$.}.

\vskip 0.15in
\noindent
In \cite{rosenbaum2021deep}, the approximation of the pricing mapping $P(\cdot)$ is learned directly from samples generated through Monte-Carlo simulations. One can also obtain $P(\cdot)$
through solving the following partial differential equation (PDE):
\begin{equation}
  \frac{\partial P}{\partial t} + \sum_{i=1}^n\frac{\partial P}{\partial Z^i}(-\gamma_iZ^i-\lambda \sum_i^nc_iZ^i) + 
                    \frac{1}{2}\sum_{i=1}^n\frac{\partial^2P}{\partial S\partial Z^i}\eta SV 
                     + \frac{1}{2}\sum_{i,j}^n\frac{\partial^2P}{\partial Z^i\partial Z^j}\eta^2V = 0 \, ,
  \label{eq:pricing_pde}                     
\end{equation}
with terminal condition $P(T, X) = g(X)$, where $g(\cdot)$ represents the payoff function of the considered derivative.
It is \textit{a priori} intricate to solve Equation (\ref{eq:pricing_pde}) with conventional methods, especially when $n$ is large. One can apply 
the method with the help of deep learning introduced in \cite{sirignano2018dgm}, which is reported to be effective on high-dimensional PDEs. For the tests presented in Section \ref{sec:num_res}, we rely on the method introduced in \cite{rosenbaum2021deep} for the computation of $\delta_t$.  


\subsection{Multi-asset market making}
\label{sec:multi_mm}
Our market making problem is considered over the period $[0, T]$, with $T$ smaller than the expiration of any asset in our selected basket. For asset $j\in\{1,\cdots,d\}$, at each point in time, the market maker decides whether to make a market at the limits $P^j_t$ plus/minus one-half tick size. Let $l^{j,b}_t, l^{j,a}_t\in\{0, 1\}$ represent 
the decisions concerning respectively the bid and ask side. The value $1$ means the participation of the market maker with a limit order of constant size $m^j$, and $0$ stands for the case where the order is not placed in the queue corresponding to the best limits so that the probability of execution
is small. We denote the two point processes modeling the number of transactions at the bid and ask size by $(N^{j,b}_t)_{t\in[0,T]}$ 
and $(N^{j,a}_t)_{t\in[0,T]}$ respectively. Then the dynamics of the inventory process $(q^j_t)_{t\in[0,T]}$ of asset $j$ is given by
\begin{equation*}
  dq^j_t = m^jdN^{j,b}_t - m^jdN^{j,a}_t \, .
\end{equation*}
We denote by $(\lambda^{j,b}_t)_{t\in[0,T]}$ and $(\lambda^{j,a}_t)_{t\in[0,T]}$ the intensity processes of $(N^{j,b}_t)_{t\in[0,T]}$ 
and $(N^{j,a}_t)_{t\in[0,T]}$ respectively, which verify
\begin{equation*}
  \lambda^{j, b}_t = l^{j,b}_t\Lambda^{j,b}\mathbbm{1}_{\{q_t^j+m^j<=Q^j\}} \, , \qquad \lambda^{j, a}_t = l^{j,a}_t\Lambda^{j,a}\mathbbm{1}_{\{q_t^j-m^j>=-Q^j\}} \, ,
\end{equation*}
where $Q^j$ stands for the maximum inventory of asset $j$ that the market maker is willing to hold, which is a multiple of $m^j$ without loss of generality. 
We do not differentiate the queue positions of the orders in this work, and denote the constant execution probability at the bid and ask side by $\Lambda^{j,b}$ and 
$\Lambda^{j, a}$ respectively. 
\vskip 0.15in
\noindent
Denoting by $D^j$ the tick size of asset $j$, the resulting dynamics of the cash process $(Y_t)_{t\in[0,T]}$ of the market maker is:
\begin{equation*}
  \begin{split}
    dY_t &= \sum_{j=1}^d -(P^j_t-\frac{D^j}{2})m^jdN^{j,b}_t + (P^j_t+\frac{D^j}{2})m^jdN^{j,a}_t \, , \\
        &= \sum_{j=1}^d \frac{D^j}{2}m^j(dN^{j,b}_t + dN^{j,a}_t) - P^j_tdq^j_t  \, .
  \end{split}
\end{equation*}
Let $(\Pi_t)_{t\in[0,T]}$ be the process representing the mark-to-market value of the market maker's portfolio, \textit{i.e.} $\Pi_t := Y_t+\sum_{j=1}^dP^j_tq^j_t$. 
Its dynamics are given by
\begin{equation*}
  \begin{split}
    d\Pi_t &= dY_t + \sum_{j=1}^dP^j_tdq^j_t + \sum_{j=1}^dq^j_tdP^j_t \, , \\
          &= \sum_{j=1}^d\frac{D^j}{2}m^j(dN^{j,b}_t + dN^{j,a}_t) + q^j_t\delta^j_tdS_t  \, .
  \end{split}
\end{equation*}
The objective function of the market maker considered here is akin to those used in \cite{cartea2017algorithmic, cartea2014buy, cartea2018algorithmic}, where
the expected terminal wealth, \textit{i.e.} $\Pi_t$, is maximized while the inventory risk is penalized. More precisely, we are interested in the following optimal control
problem:
\begin{equation}
  \sup_{\substack{l^{1,b} \cdots, l^{d,b} \\ l^{1,a},\cdots,l^{d,a}}} \mathbb{E}[\Pi_T  - \sum_{j=1}^d\frac{\kappa_j}{2}\int_0^T\sigma_t^2(q^j_t\delta^j_t)^2dt - \frac{\kappa}{2}\int_0^T\sigma_t^2(\sum_{j=1}^dq^j_t\delta^j_t)^2dt] \, ,
  \label{eq:obj_func}
\end{equation}
with $\sigma_t=S_t\sqrt{V_t}$, $(\kappa_j)_{j=1,\cdots, d}$, and $\kappa$ nonnegative constants. The two items with integration represent the quadratic variation of 
the inventory risk of individual assets and the portfolio respectively. The market maker can thus consider only the whole portfolio's  net inventory risk ($\kappa_j=0, j=1,\cdots,d$) or 
express also some specific views on the inventory risk of certain particular assets ($\kappa_j\neq 0, j\in\{1,\cdots,d\}$). Of course, other objective functions can be thought of, 
such as the widely used constant absolute risk aversion utility function as in \cite{avellaneda2008high}. We refer to \cite{gueant2017optimal} for a discussion on the equivalence between these 
two classes of objective functions. Now we can define the value function
\begin{equation*}
  u:(t, S, Z, q)\in[0,T] \times \mathbb{R}_{+} \times \mathbb{R}^n \times \mathcal{Q} \mapsto u(t, S, Z, q)
\end{equation*}
associated with (\ref{eq:obj_func}) as
\begin{equation*}
  \begin{split}
    u(t, S, Z, q) = \sup_{\big(\substack{l^{1,b}_s \cdots, l^{d,b}_s \\ l^{1,a}_s,\cdots,l^{d,a}_s}\big)_{s\in[t,T]}}
                    \mathbb{E}_{t, S, Z, q}\Big[&\int_t^T \sum_{j=1}^d (\sum_{k=a,b}\frac{D^j}{2}m^jl_s^{j,k}\Lambda^{j,k}
                        \mathbbm{1}_{\{q^j_s+\phi^km^j\in\mathcal{Q}^j\}}ds + q^j_s\delta^j_sdS_s)  \\
                     &  -\sum_{j=1}^d\frac{\kappa_j}{2}\int_t^T\sigma_s^2(q^j_s\delta^j_s)^2ds - \frac{\kappa}{2}\int_t^T\sigma_s^2(\sum_{j=1}^dq^j_s\delta^j_s)^2ds 
                        \Big] \, ,
  \end{split}
\end{equation*}
where $Z:=(Z^1,\cdots,Z^n)$, $q:=(q^1,\cdots,q^d)$, $\mathcal{Q}^j=\{-Q^j, -Q^j+m^j, \cdots, Q^j\}$, $\mathcal{Q}=\prod_{j=1}^d\mathcal{Q}^j$ and
\begin{equation*}
  \phi^j = \begin{cases}
        &+1 \quad \text{if} \quad j=b \, , \\
        &-1 \quad \text{if} \quad j=a \, .
        \end{cases}
\end{equation*}
Therefore, $u$ has $(2+n+d)$ variables. We introduce in the following some approximations to bypass this curse of dimensionality. 

\section{Value function approximation}
\label{sec:prob_approx}

\subsection{The Hamilton-Jacobi-Bellman equations}
To reduce the dimensionality coming from the multidimensional nature of the QRH model, we suggest the following approximation:

\begin{approximation}
\label{app:const_z}
We approximate $(Z^i_t)_{t\in[0,T], i=1,\cdots,n}$ and $(\delta^j_t)_{t\in[0,T], j=1,\cdots,d}$ by its initial value, \textit{i.e.}
\begin{equation*}
  \begin{split}
      Z^i_t &= Z^i_0  \quad i=1,\cdots,n \, , \\
      \delta^j_t &= \delta^j_0 =: \delta^j \quad j=1,\cdots,d \, ,
  \end{split}
\end{equation*}
for $t\in[0,T]$. Then we have $V_t=V_0$ for $t\in[0,T]$ by (\ref{eq:qrh_model}). We replace the dynamic of $(S_t)_{t\in[0,T]}$ by
  \begin{equation*}
    dS_t = \mu dt + \sigma dW_t, \quad t\in[0, T] \, ,
  \end{equation*}
where $\sigma := S_0\sqrt{V_0}$ and we allow a constant drift $\mu$.
\end{approximation}

\vskip 0.15in
\noindent
The above approximation is relevant when the time horizon of the problem is relatively short. 
In \cite{rosenbaum2021deep} the effectiveness of daily hedging under the QRH model is illustrated through numerical tests on simulated and market data.
We give additional daily hedging results in the next section. All these examples indicate that it is reasonable to consider $\delta$ constant at the daily scale.    
In practice, the market maker can always reset the algorithm with updated parameters in the case of significant market movement. We recall that the computation of $(\delta^j_0)_{j=1,\dots,d}$ and $V_0$ at the beginning of the day is ruled by the QRH model, see \cite{rosenbaum2021deep} for more details. Now the value function satisfies
\begin{equation*}
  \forall (t, S, Z, q)\in [0,T] \times \mathbb{R}_{+} \times \mathbb{R}^n \times \mathcal{Q}, u(t, S, Z, q) = v(t, q) \, ,
\end{equation*}
where 
\begin{equation}
  \begin{split}
      v(t, q) = \sup_{\big(\substack{l^{1,b}_s \cdots, l^{d,b}_s \\ l^{1,a}_s,\cdots,l^{d,a}_s}\big)_{s\in[t,T]}} \mathbb{E}_{t,q}\big[&
          \int_t^T\sum_{j=1}^d(\sum_{k=a,b}\frac{D^j}{2}m^jl_s^{j,k}\Lambda^{j,k}\mathbbm{1}_{\{q_s^j+\phi^km^j\in\mathcal{Q}^j\}} + q^j_s\delta^j\mu)ds \\
          & -\sum_{j=1}^d\frac{\kappa_j}{2}\sigma^2\int_t^T(q^j_s\delta^j_s)^2ds - \frac{\kappa}{2}\sigma^2\int_t^T(\sum_{j=1}^dq^j_s\delta^j)^2ds
      \big] \, .
  \end{split}
  \label{eq:value_func_v}
\end{equation}
Now the value function $u(\cdot)$ is transformed to one with $(1+d)$ variables. Following \cite{oksendal2005stochastic}, the HJB equation associated with (\ref{eq:value_func_v})
is given by a system of ODEs verifying
\begin{equation}
    \begin{split}
        0 = &-\frac{\partial v}{\partial t} - \mu\sum_{j=1}^dq^j\delta^j+\sum_{j=1}^d\frac{\kappa_j}{2}\sigma^2(q^j\delta^j)^2 + \frac{\kappa}{2}\sigma^2(\sum_{j=1}^dq^j\delta^j)^2  \\
        &-\sum_{j=1}^d\sum_{k=a,b}\mathbbm{1}_{\{q^j + \phi^kz^j\in\mathcal{Q}^j\}} \sup_{l^{j,k}\in\{0,1\}}l^{j,k}\Lambda^{j,k}(m^j\frac{D^j}{2} + v(t, q+\phi^km^je^j) - v(t, q)) \, ,
    \end{split}
    \label{eq:hjb}
\end{equation}
with terminal condition
\begin{equation}
  v(T, q) = 0 \, ,
  \label{eq:hjb_terminal}
\end{equation}
where $\{e^j\}_{j=1}^d$ is the canonical basis of $\mathbb{R}^d$. We can rewrite the terms involving the control variables as 
\begin{equation*}
  \begin{split}
      \mathbbm{1}_{\{q^j + \phi^km^j\in\mathcal{Q}^j\}}&\sup_{l^{j,k}\in\{0,1\}}l^{j,k}\Lambda^{j,k}(m^j\frac{D^j}{2} + v(t, q+\phi^km^je^j) - v(t, q)) \\
     =: &  \mathbbm{1}_{\{q^j + \phi^kz^j\in\mathcal{Q}^j\}}m^jH^{j,k}(\frac{v(t,q) - v(t, q+\phi^km^je^j)}{m^j}) \, ,
  \end{split}
\end{equation*}
with
\begin{equation}
  H^{j,k}(p) = \Lambda^{j,k}\mathbbm{1}_{\{p \leq \frac{D^j}{2}\}}(\frac{D^j}{2} - p )  \\  \, .
  \label{eq:hamiltonian}
\end{equation}

\vskip 0.15in
\noindent
In the following we give the results on the existence and uniqueness of a solution of (\ref{eq:hjb}) with terminal 
condition (\ref{eq:hjb_terminal}). We give the sketch of the proof as it follows the same steps as in 
\cite{gueant2017optimal}. 
\begin{theorem}
  There exists a unique solution $v:[0,T]\times \mathcal{Q} \mapsto \mathbb{R}$, $C^1$ in time, solution of Equation (\ref{eq:hjb}) with terminal condition (\ref{eq:hjb_terminal}).
\end{theorem}
\begin{proof}
   It is clear that $H^{j,k}$ are Lipschitz continuous functions for all $j\in\{1,\cdots,d\}$ and $k\in\{a,b\}$. Equation (\ref{eq:hjb}) with terminal condition (\ref{eq:hjb_terminal})
   can be viewed as a backward Cauchy problem. According to the Cauchy-Lipschitz theorem, there exists some $\tau\in[0, T)$ such that the Equation (\ref{eq:hjb}) with the 
   terminal condition (\ref{eq:hjb_terminal}) has a unique solution, $C^1$ in time, on the interval $(\tau, T)$.
   \vskip 0.1in
   \noindent
   $\forall q\in\mathcal{Q}$, it is clear that the function $t\in(\tau, T) \mapsto v(t, q) + (-\mu\sum_{j=1}^dq^j\delta^j + \sum_{j=1}^d\frac{\kappa_j}{2}\sigma^2(q^j\delta^j)^2 + \frac{\kappa}{2}\sigma^2(\sum_{j=1}^dq^j\delta^j)^2)(T-t)$ 
   is a decreasing function. Therefore, the only reason why there would not be a global solution on $[0, T]$ is because $\sup_{q\in\mathcal{Q}}v(t,q)$ blows up at $\tau>0$. By the fact 
   that $H^{j,k}$ is decreasing, a classical comparison principle can be obtained. It is clear that $\bar{v}(t,q)=\sum_{j=1}^d(H^{j,b}(0) + H^{j,a}(0) + |\mu\delta^j|Q^j)(T-t)$ defines a supersolution 
   of (\ref{eq:hjb}) with terminal condition (\ref{eq:hjb_terminal}). Then by the comparison principle, we have
   \begin{equation*}
      \sup_{q\in\mathcal{Q}}v(t, q) \leq \sum_{j=1}^d(H^{j,b}(0) + H^{j,a}(0) + |\mu\delta^j|Q^j)(T-t) \, .
   \end{equation*}
   This shows that $v$ is bounded for $t\in[0, T]$, which leads to its existence and uniqueness on $[0, T]\times \mathcal{Q}$.  
\end{proof}
\vskip 0.15in
\noindent
By a classical verification argument, we can get the following result:
\begin{theorem}
  Considering the solution $v$ of Equation (\ref{eq:hjb}) with terminal condition (\ref{eq:hjb_terminal}), the optimal market making decisions in 
  the problem (\ref{eq:value_func_v}) are given by
  \begin{equation*}
    l^{j,k\ast}_t(q)= \mathbbm{1}_{\{q^j_{t-} + \phi^km^j\in\mathcal{Q}^j, \frac{v(t, q_{t-}) - v(t, q_{t-}+\phi^km^je^j)}{m^j}\leq \frac{D^j}{2}\}} \, ,\quad \forall j\in\{1,\cdots, d\}, k\in\{a,b\} \, .
  \end{equation*}
  \label{theo:control}
\end{theorem}

\begin{remark}
  \label{re:port_risk}
  When the market maker controls only the portfolio's net risk, \textit{i.e.} $\forall j\in\{1,\cdots, d\}, \kappa_j =0$ and $Q^j\rightarrow +\infty$, the $d$-dimensional variable $q$ can be
  summarized with one variable. More precisely, by introducing the maximum risk bound $R>0$ and $r_t:=\sum_{j=1}^dq^j_t\delta^j$, now the value function $\theta:(t, r)\in[0,T]\times[-R, R] \mapsto \theta(t,R)$
  associated with the problem can be written as
  \begin{equation}
      \theta(t, r) = \sup_{\big(\substack{l^{1,b}_s \cdots, l^{d,b}_s \\ l^{1,a}_s,\cdots,l^{d,a}_s}\big)_{s\in[t,T]}} 
      \mathbb{E}_{t,r}\big[\int_t^T\sum_{j=1}^d\sum_{k=a,b}\frac{D^j}{2}m^jl_s^{j,k}\Lambda^{j,k}\mathbbm{1}_{\{r_s+\phi^km^j\delta^j\in[-R, R]\}} 
      + \mu r_s ds - \frac{\kappa}{2}\sigma^2\int_t^T(r_s)^2ds \big] \, .
  \label{eq:value_func_theta}
  \end{equation}
And the associated HJB equation is given by
\begin{equation*}
  \begin{split}
    0 = & -\frac{\partial\theta}{\partial t} -\mu r + \frac{\kappa}{2}\sigma^2r^2   \\
      &-\sum_{j=1}^d\sum_{k=a,b}\mathbbm{1}_{\{r + \phi^km^j\delta^j\in\mathcal{R}\}}m^jH^{j,k}(\frac{\theta(t,r) - \theta(t, r+\phi^km^j\delta^j)}{m^j}) \, ,
  \end{split}
  \label{eq:hjb_port_risk}
\end{equation*}
with terminal condition $\theta(T, r) = 0$.
\end{remark}

\noindent

\subsection{Quadratic approximation}
In our case with a portfolio made of SPX and numerous derivatives, Equation (\ref{eq:hjb}) becomes intricate to solve with classical 
numerical methods. Here we follow the idea of \cite{bergault2021closed} that is approximating the Hamiltonian functions $H^{j,k}(\cdot)$ 
with quadratic ones. Then closed-form solution can be obtained for the new HJB equation. Moreover, the asymptotic of the solution can be derived, which are helpful for practical use. Note that most of the development in this part follows the same methodologies as in \cite{bergault2021closed}. We give some key results in the following to avoid any ambiguity.

\vskip 0.15in
\noindent
The market making problem considered in \cite{bergault2021closed} is more adapted to over-the-counter market, where the optimal
control is the quoting price and is taken to be continuous. The resulting Hamiltonian functions are $C^2$ and a natural choice is to 
use Taylor expansions as the approximations. Clearly this do not apply for (\ref{eq:hamiltonian}). In this work we propose the 
following choice as the approximated Hamiltonian functions:
\begin{equation}
   \hat{H}^{j,k}(p) = \Lambda^{j,k}(\frac{\alpha^j}{2}p^2-p+\frac{D^j}{2}),
   \label{eq:hamiltonian_approx}
\end{equation}
with $\alpha^j > 0$. Actually we have $\hat{H}^{j,k}(0) = H^{j,k}(0)$ and $(\hat{H}^{j,k})^{\prime}(0)= (H^{j,k})^{\prime}(0)$ 
for any $(j,k)\in\{1,\cdots,d\}\times\{a,b\}$. We can choose $\alpha^j\sim\frac{1}{D^j}$ to make $\hat{H}^{j,\cdot}$ close to $H^{j,\cdot}$
around $0$. In the case without
maximal inventory constraints for individual assets, \textit{i.e.} $Q^j=+\infty, j\in\{1,\cdots,d\}$, following the same computations 
as in \cite{bergault2021closed},  the approximation of $v$ associated with $(\hat{H}^{j,k})_{j\in\{1,\cdots,d\}, k\in\{a,b\}}$, 
denoted by $\hat{v}$, can be written as a quadratic function of $q$, \textit{i.e.}
\begin{equation*}
  \hat{v}(t, q) = -q^TA(t)q-q^TB(t)-C(t) \, ,
\end{equation*}
where $A:[0,T]\mapsto S_d^+, B:[0,T]\mapsto\mathbb{R}^d$ and $C:[0,T]\mapsto \mathbb{R}$ are given by the unique solution of 
a system of ODEs\footnote{We denote by $S_d^+$ the set of positive semi-definite symmetric matrices of dimension $d\times d$.}. 
Since we are more interested in the behavior of $\hat{v}$ as $T\rightarrow +\infty$, 
we focus mostly on the asymptotic formulas of $A$ and $B$\footnote{We ignore $C$ since it is irrelevant for the dependence of the value function 
on the states.}. Let $\delta:=(\delta^1, \cdots, \delta^d)^T$ and $\Sigma:=diag(\frac{\kappa_1(\delta^1)^2}{\kappa},\cdots, \frac{\kappa_d(\delta^d)^2}{\kappa}) + \delta\delta^T$.
We have 
\begin{equation*}
  \begin{split}
      A(0)  &\overset{T\rightarrow+\infty}{\rightarrow}  \frac{\sigma}{2}\sqrt{\kappa}\Gamma \, , \\
      B(0)  &\overset{T\rightarrow+\infty}{\rightarrow} - D_{+}^{-\frac{1}{2}}\hat{A}^+D_+^{\frac{1}{2}}\delta\mu - D_+^{-\frac{1}{2}}\hat{A}\hat{A}^+D_+^{-\frac{1}{2}}(V_- + \frac{\sigma}{2}\sqrt{\kappa}D_- \mathcal{D}(\Gamma))\, ,
  \end{split}
\end{equation*}
where
\begin{equation*}
  \begin{split}
    D_+ &= diag\big((\Lambda^{1,b}+\Lambda^{1,a})\alpha^1z^1, \cdots, (\Lambda^{d,b}+\Lambda^{d,a})\alpha^dz^d\big) \, ,\\
    D_- &= diag\big((\Lambda^{1,b}-\Lambda^{1,a})\alpha^1z^1, \cdots, (\Lambda^{d,b}-\Lambda^{d,a})\alpha^dz^d\big) \, , \\
    V_- &= \big( (-\Lambda^{1,b}+\Lambda^{1,a})z^1, \cdots, (-\Lambda^{d,b}+\Lambda^{d,a})z^d\big) \, , \\
    \Gamma &= D^{-\frac{1}{2}}_+(D^{\frac{1}{2}}_+\Sigma D^{\frac{1}{2}}_+)^{\frac{1}{2}}D^{-\frac{1}{2}}_+\, , \quad 
    \hat{A} = \sqrt{\kappa V}(D^{\frac{1}{2}}_+\Sigma D^{\frac{1}{2}}_+)^{\frac{1}{2}} \, ,
  \end{split}
\end{equation*}
$\hat{A}^+$ is the Moore-Penrose inverse of $\hat{A}$ and $\mathcal{D}$ is the linear operator mapping a matrix onto the vector of its diagonal coefficients.
As we focus on the asymptotic of the approximated solution in the following, we use $A, B$ and $\hat{v}(q)$ to represent the respective asymptotic formulas for ease of notation.
As suggested in \cite{bergault2021closed}, for any $(j,k)\in\{1,\cdots, d\}\times\{a, b\}$, one important application of $\hat{v}$
is to make greedy decisions following
\begin{equation*}
  \begin{split}
    \tilde{l}^{j,k\ast}(q) &= \mathbbm{1}_{\{q^j+\phi^km^j\in\mathcal{Q}^j, \frac{\hat{v}(q)-\hat{v}(q_+\phi^km^je^j)}{m^j}\leq \frac{D^j}{2}\}} \\
    &= \mathbbm{1}_{\{q^j+\phi^km^j\in\mathcal{Q}^j, 2\phi^k(e^j)^TAq + m^j(e^j)^TAe^j + \phi^k(e^j)^TB \leq \frac{D^j}{2}\}}\, . 
  \end{split}
\end{equation*}
It is clear that the set $\{q|\tilde{l}^{j,k\ast}=1, q\in\mathcal{Q}\}$ and $\{q|\tilde{l}^{j,k\ast}=0, q\in\mathcal{Q}\}$ are separated by an affine hyperplane defined on $\mathcal{Q}$. 
In the particular case with only portfolio-level net risk controlled, as suggested in Remark \ref{re:port_risk}, and symmetric order execution intensities on the bid and ask sides, 
\textit{i.e.} $\Lambda^{j,a}=\Lambda^{j,b}=:\Lambda^j$ for all $j\in\{0,\cdots,d\}$, applying the above quadratic approximation on Equation (\ref{eq:hjb_port_risk}), the optimal decisions 
can be simplified as 
\begin{equation*}
  \tilde{l}^{j, k\ast}(r) = \mathbbm{1}_{\{r+ \phi^km^j\delta^j\in[-R, R], 2Dr\delta^j + (\delta^j)^2m^j \leq \frac{D^j}{2}\}} \, ,
\end{equation*}
where
\begin{equation*}
  D = \frac{\sigma}{2}\sqrt{\frac{\kappa}{2\sum_j^d(\delta^j)^2\Lambda^jm^j\alpha^j}} \, .
\end{equation*}

\begin{remark}
  In the market making problem introduced in Section \ref{sec:multi_mm}, the agent is restricted to sending limit orders only on the first bid and ask limits.
  The problem can be easily extended to cover other limits. The resulting Hamiltonian functions can also be approximated and then a closed-form asymptotic solution can be obtained 
  in a similar manner as the above.
\end{remark}
 

\section{Numerical results}
\label{sec:num_res}

\subsection{Daily hedging with SPX}
In this part, we perform discrete hedging to evaluate the relevance of the QRH model and the idea of using constant $\delta_t$ during a short period, as suggested in Approximation \ref{app:const_z}. 
For asset $j$, the cumulative profit of discrete hedging with a time step $\Delta t$ is given by
\begin{equation*}
  \mathcal{J}_t^{j,\delta} = \sum_{k=0}^{\lfloor t/\Delta t \rfloor - 1}\hat{\delta}_{t_k}^j\Delta S_{t_{k+1}} + \hat{\delta}^j_{t_{\lfloor t/\Delta t\rfloor}}(S_t - S_{\lfloor t/\Delta t\rfloor \Delta t}) \, ,    
\end{equation*}
where $t_k=k\Delta t$, $\Delta S_{t_k}:=S_{t_k} - S_{t_{k-1}}$, and $\hat{\delta}_{t_k}^j$ is the hedging quantity at $t_k$ computed by neural networks. 
We compare $\mathcal{J}_t^{j,\delta}$ with the price evolution of asset $j$ observed in the market, defined by
\begin{equation*}
  \mathcal{J}_t^{j,P} = P^j_t - P^j_0 \, .
\end{equation*}
Figure \ref{fig:hedge_exp_1} and \ref{fig:hedge_exp_2} show the results of two experiments of daily hedging on SPX call, VIX call and VIX future with the following characteristics:
\begin{enumerate}
  \item SPX call: maturity 2017-07-21, strike 2460; VIX call: maturity 2017-07-19, strike 11.5; VIX future: maturity 2017-07-19.
  \item SPX call: maturity 2017-10-20, strike 2550; VIX call: maturity 2017-10-18, strike 13; VIX future: maturity 2017-10-18.
\end{enumerate}
On day zero of each experiment, we calibrate the QRH model with market data and compute $(\hat{\delta}_{t_0}^j)_{j=1,\dots,d}$. For each following day, $(Z^i)_{i=1,\dots,n}$ are updated with respect to the evolution of SPX, and then $(\hat{\delta}^j)_{j=1,\dots,d}$ are recomputed. Remarkably in both tests, the hedging portfolios can follow well the market price's evolution with daily rebalancing for all assets. 
It implies that in practice high-frequency hedging is not mandatory to obtain satisfactory results, and verifies 
the rationality of Approximation \ref{app:const_z} over the time horizon of the considered market making problem, which is usually shorter than one day.   

\begin{figure}[!h]
  \centering
  \includegraphics[width=\textwidth]{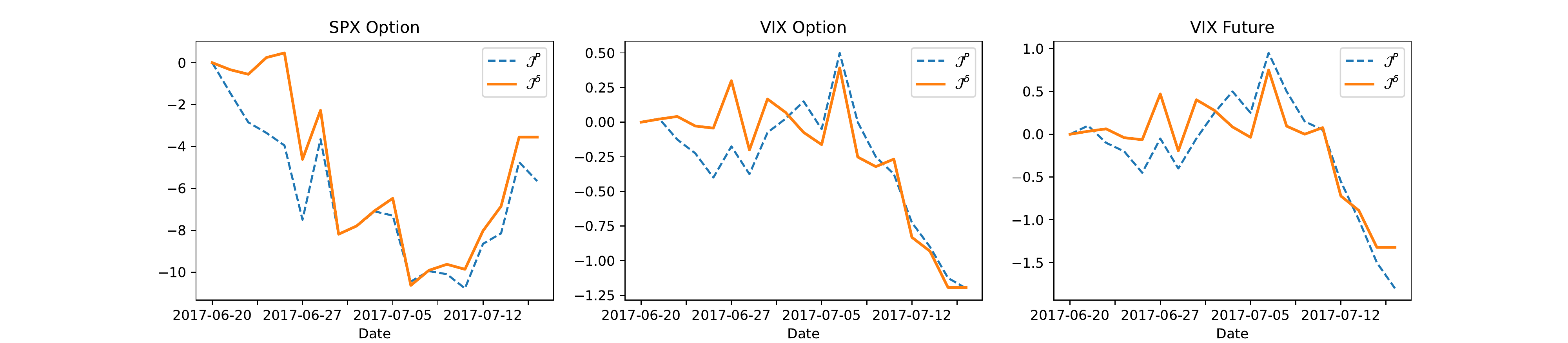}
  \caption{Daily hedging example 1.}
  \label{fig:hedge_exp_1}
\end{figure}

\begin{figure}[!h]
  \centering
  \includegraphics[width=\textwidth]{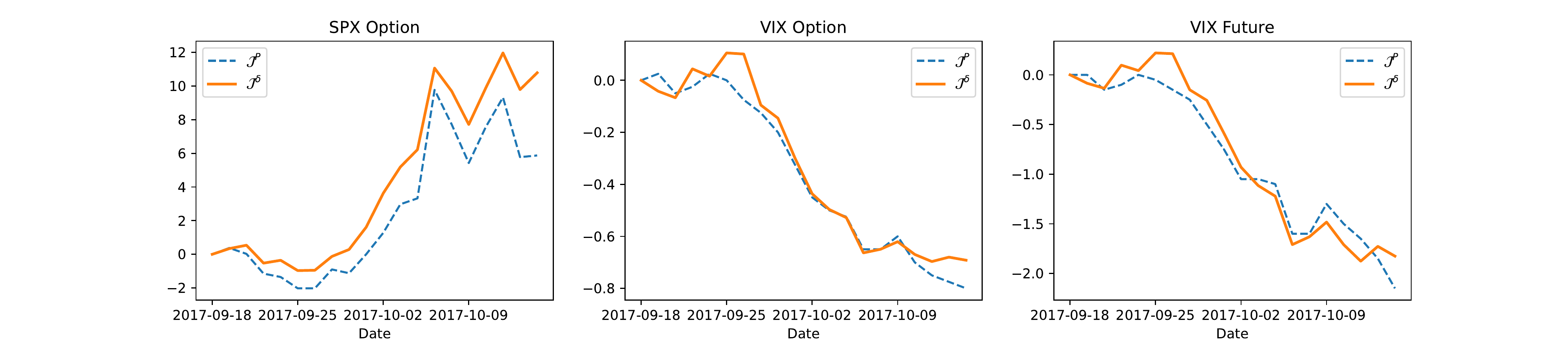}
  \caption{Daily hedging example 2.}
  \label{fig:hedge_exp_2}
\end{figure}

\subsection{Market making approximation}
\subsubsection{Example 1}
In this part, we assess the proposed quadratic approximation. We begin with a two-asset example with SPX and VIX future. 
The parameters of the QRH model are taken as follows\footnote{The same multi-factor approximation method as that in \cite{rosenbaum2021deep} is used, leading to the same $(c_i)_{i=1,\cdots,10}$ and $(\gamma_i)_{i=1,\cdots,10}$}:
\begin{equation*}
  \begin{split}
        \lambda &= 1.7, \quad \eta=1.5, \quad a=0.265, \quad b=0.246, \quad c=0.0001,   \\
        Z_0 &= (-0.009,  0.015,  0.011,  0.036,  0.002, -0.011, -0.018,  0.074, 0.142, -0.171) \, .
  \end{split}
\end{equation*}
Then we have $V_0 = 0.18\, \text{year}^{-1}$. We set $S_0=3000$\$, and $\mu_0=0\, \text{year}^{-1}$. 
The maturity of VIX future is set to be one month. Along with $\delta^1=1$, using the neural network approximating the pricing function of VIX future, 
we have $\delta^2 = -0.028$.

\vskip 0.15in
\noindent
As for the market making related parameters, we consider the following values:
\begin{itemize}
  \item Tick size: $D^1 = 0.25\$, D^2 = 0.05\$\,.$
  \item Order size: $m^1 = 1, m^2 = 20\,.$
  \item Maximum inventory: $Q^1 = 15, Q^2=300\,.$
  \item Execution intensity: $\Lambda^{1,b}=\Lambda^{1,a}=1\,\text{seconds}^{-1}, \Lambda^{2,b}=\Lambda^{2,a}=0.1\,\text{seconds} \,.$
  \item Risk penalization: $\kappa^1=0.005\$^{-1} , \kappa^2=0.005\$^{-1}, \kappa= 0.01\$^{-1}\,.$
  \item Time horizon: $T=300 \,\text{seconds}$.
\end{itemize}
Here some characteristics of SPX are referred to the front month E-mini S\&P 500 future. 
In real life, the contract multiplier of E-mini S\&P 500 future is 50, and that of VIX future is 1000.
Here we normalize them against 50 to have order sizes. The execution intensities are 
selected such that on average the volume of SPX is 10 times that of VIX future. 

\vskip 0.15in
\noindent
First, we use a classical Euler scheme to approximate numerically 
the true value function $v(t, q)$. Figure \ref{fig:value_func_two_euler} shows the results for several $t\in[0, T]$.  
Interestingly, the shapes of $v(t, q)$ at different $t$ are very similar, implying that time
does not play a crucial role for 
the optimal controls according to Theorem \ref{theo:control}. This is confirmed by Figure \ref{fig:optim_ctr_two_euler}, 
showing that for SPX or VIX future, the set $\{q|l^{\cdot,b\ast}=0\}$ does not change a lot with $t$ when $t$ is relatively small.

\begin{figure}[!h]
  \centering
  \includegraphics[width=.8\textwidth]{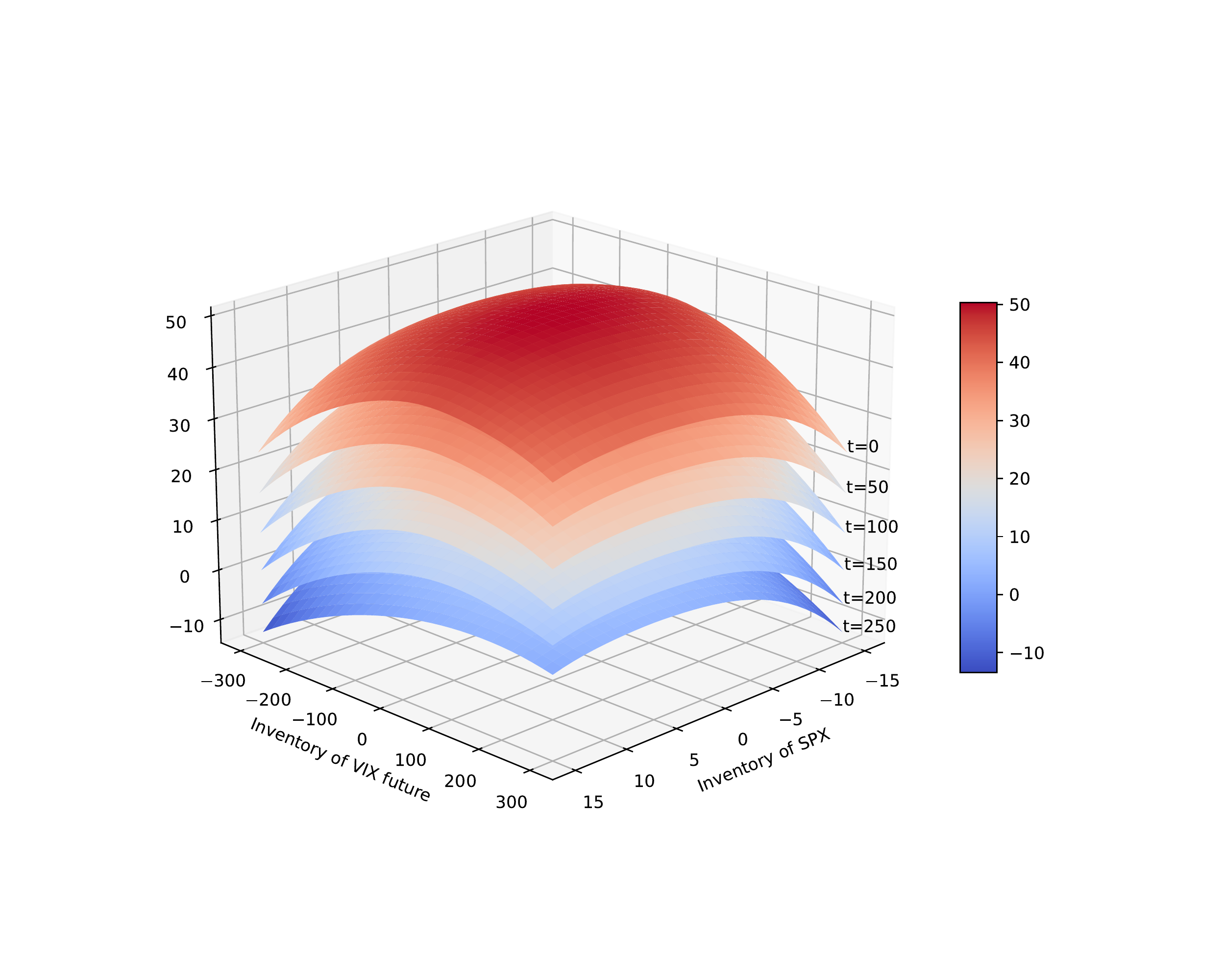}
  \caption{Snapshots of the value function $v(t, q)$ for t=0, 50, 100, 150, 200, 250.}
  \label{fig:value_func_two_euler}
\end{figure}

\begin{figure}[!h]
  \centering
  \includegraphics[width=\textwidth]{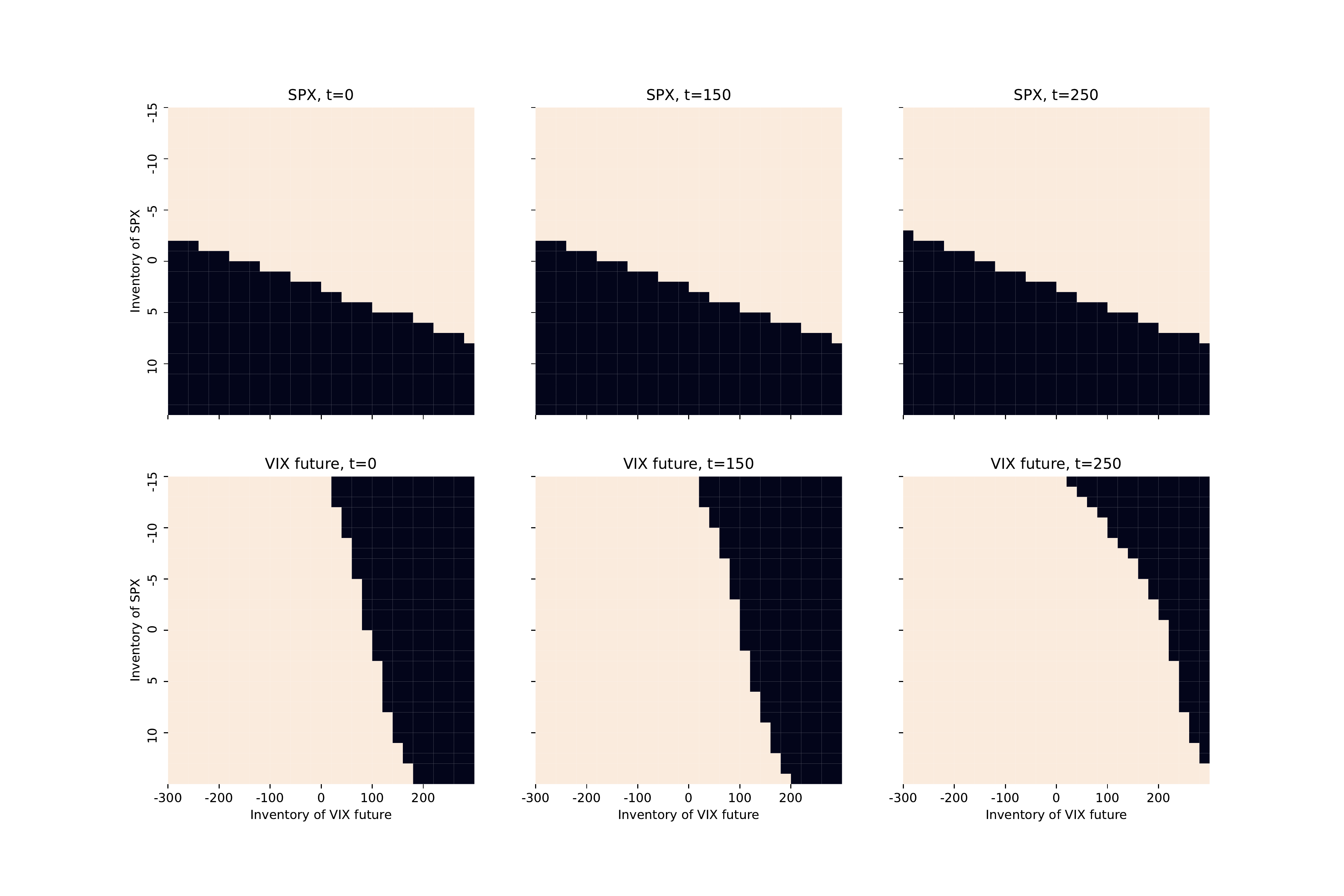}
  \caption{Optimal decisions on whether sending bid order for SPX and VIX future at different $t$ and $q$. 
     The black areas stand for the set $\{q|l^{j,b\ast}_t=0\}$.}
  \label{fig:optim_ctr_two_euler}
\end{figure}

\vskip 0.15in
\noindent
Now we check the relevance of our quadratic approximation. As for the parameter $\alpha^j$ defined in Equation (\ref{eq:hamiltonian_approx}),
we choose $\alpha^j=\frac{1}{4D^j}, j\in\{1,\cdots,d\}$. The asymptotic approximated value function $\hat{v}(q)$ is shown in Figure \ref{fig:value_func_two_approx}. 
Its shape is very close to that of $v(0,q)$, and the resulting greedy decision $\tilde{l}^{j,k\ast}(q)$, 
given in Figure \ref{fig:optim_action_two_approx}, is very similar to $l^{j,k\ast}_0(q)$ shown in Figure \ref{fig:optim_ctr_two_euler}.

\begin{figure}[!h]
  \centering
  \includegraphics[width=\textwidth]{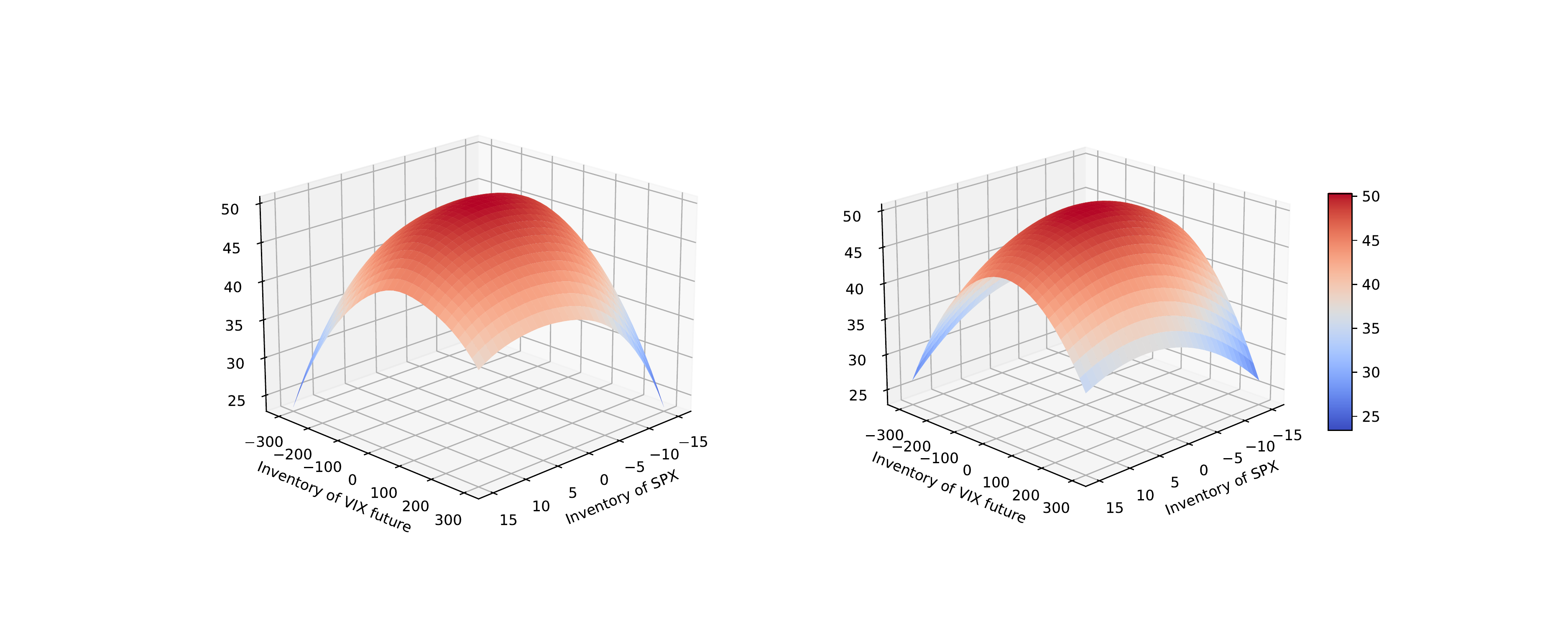}
  \caption{The asymptotic value function $\hat{v}(0, q)$ (right) as $T\rightarrow+\infty$, compared to $v(0, q)$ (left). Note that for ease of comparison 
        $C(0)$ is chosen to have $\max_q\hat{v}(0,q)=\max_qv(0,q)$.}
  \label{fig:value_func_two_approx}
\end{figure}

\begin{figure}[!h]
  \centering
  \includegraphics[width=\textwidth]{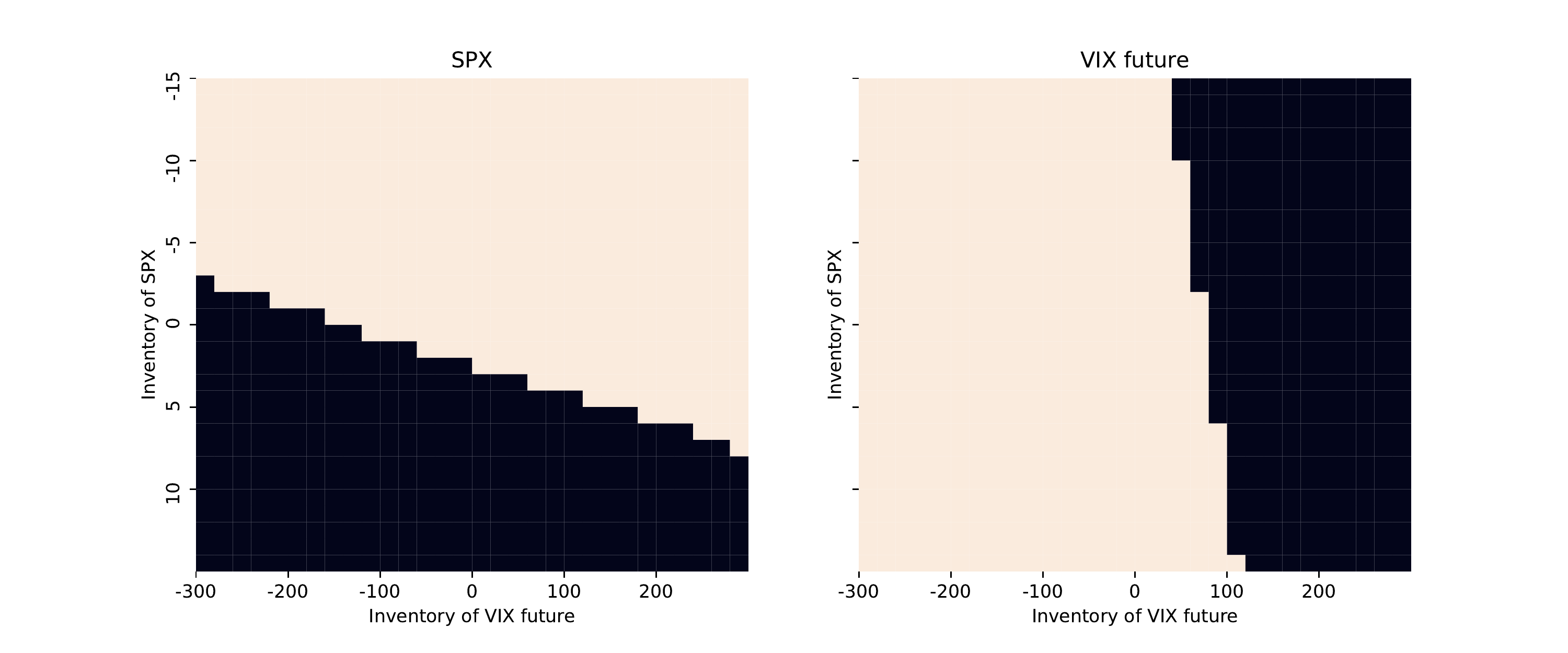}
  \caption{The greedy decisions $\tilde{l}^{j,b\ast}$ obtained with the asymptotic value function $\hat{v}(0,q)$. The black areas stand for the set $\{q|\hat{l}^{j,b\ast}=0\}$.}
  \label{fig:optim_action_two_approx}
\end{figure}

\vskip 0.15in
\noindent
To see the impact of the quadratic approximation in cases with other values of $(\kappa^i)_{j=1,\cdots,d}$ and $\kappa$, and to quantify the deviation of the market making decisions 
following the approximated solution 
from the ones based on the true solution, we conduct backtests on simulated data with different risk aversion preferences\footnote{Our simulations use neural networks to compute efficiently the price of SPX and VIX option, and VIX Future as introduced in \cite{rosenbaum2021deep}.}. 
More precisely, we vary the inventory penalization parameter $\kappa$ from $10^{-4}$ to $1$ and fix $\kappa^1=\kappa^2=\frac{\kappa}{2}$ for simplicity. 
We test two strategies $\beta^{l}$ and $\beta^{\tilde{l}}$, based respectively on $l^{\ast}$ and $\tilde{l}^{\ast}$\footnote{For ease of notation, we ignore the superscript $\ast$ in the following.}.
For each pair $(\kappa, \beta^{l})$ and $(\kappa, \beta^{\tilde{l}})$, we evaluate the mean and standard deviation of the profit of market making, 
defined as $\Pi_{T_b} - \Pi_0$ where $T_b$ stands for the backtesting horizon, on $N$ simulated paths. We use $T_b=\frac{T}{2}=150$ seconds as the backtesting horizon since we are more interested in the results given by the asymptotic of $v$. 
With $N=5000$, Figure \ref{fig:two_backtest} shows the resulting mean-risk profiles for $\beta^{l}$ and $\beta^{\tilde{l}}$.
Even though for a given $\kappa$, slight deviations from $\beta^{l}$ are observed for $\beta^{\tilde{l}}$, the global mean-risk profile of its performance is not significantly 
degraded from that of $\beta^l$. This suggests the effectiveness of the proposed approximation method.   
\begin{figure}[!h]
  \centering
  \includegraphics[width=.8\textwidth]{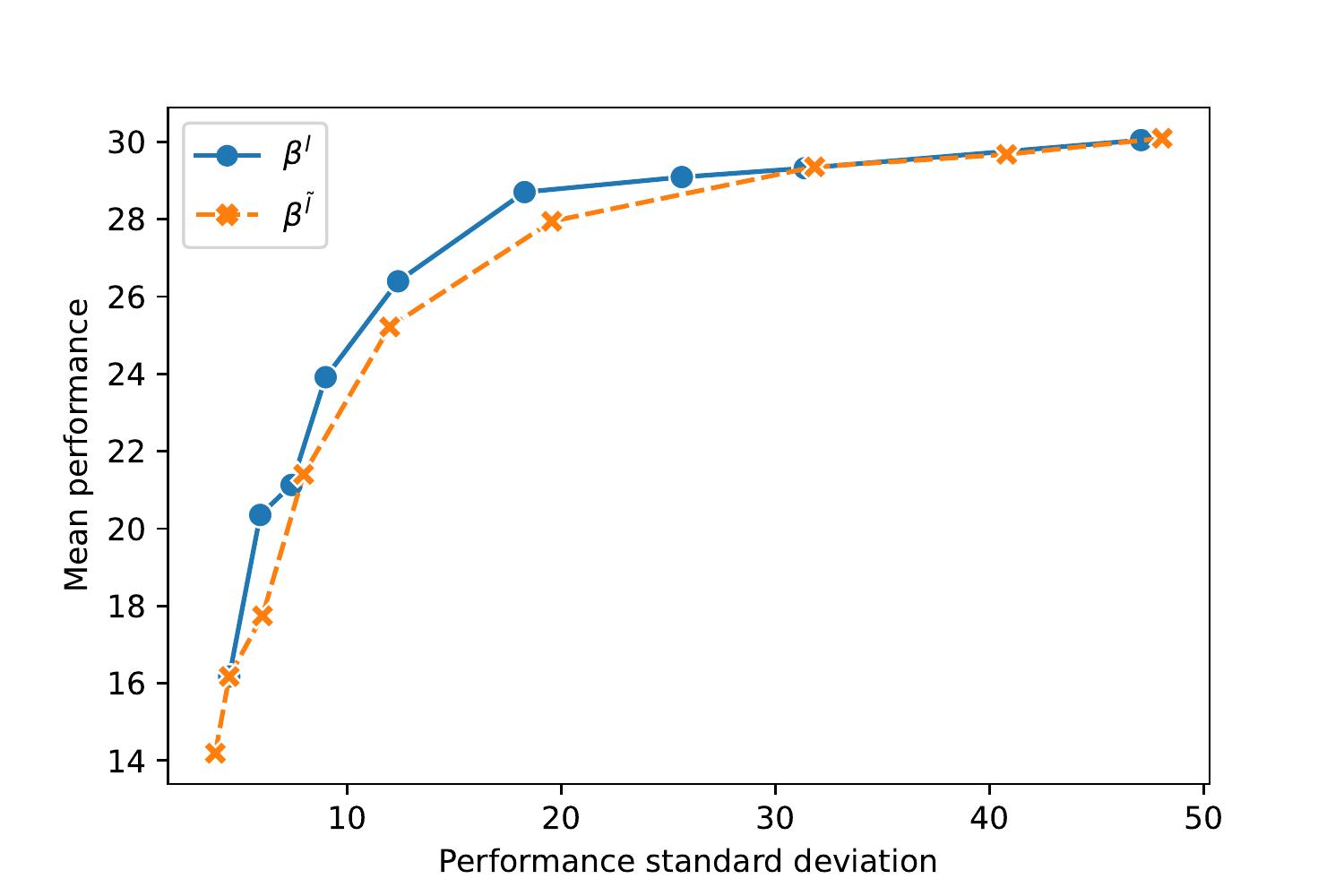}
  \caption{Mean performance and risk of the market making strategies $\beta^l$ and $\beta^{\tilde{l}}$ with varying risk penalization parameter $\kappa$ on two assets.}
  \label{fig:two_backtest}
\end{figure}

\subsubsection{Example 2}
Here we repeat the same experiment as the above for the case where only the portfolio's net inventory risk is controlled, as discussed in Remark (\ref{re:port_risk}). We consider the following assets:
\begin{itemize}
  \item SPX option with maturity $10$ days and strike $2950\$$.
  \item SPX option with maturity $25$ days and strike $3050\$$.
  \item VIX option with maturity $10$ days and strike $22\$$.
  \item VIX option with maturity $25$ days and strike $18\$$.
\end{itemize}
And we have the following parameters:
\begin{itemize}
  \item Tick size: $D^1=0.1\$, D^2=0.1\$, D^3=0.05\$, D^4=0.05\$$.
  \item Inventory risk: $\delta^1=0.533, \delta^2=0.134, \delta^3=-0.014, \delta^4=-0.013$.
  \item Order size: $m^j=2$ for $j\in\{1,2,3,4\}$.
  \item Inventory risk bound: $R=10$.
  \item Execution intensity: $\Lambda^{j,b}=\Lambda^{j,a}=0.05\text{ second}^{-1}$ for $j\in\{1,2,3,4\}$.
  \item Risk penalization: $\kappa=1$.
  \item Time horizon: $T=600$ seconds.
\end{itemize}

\vskip 0.15in
\noindent
With four assets under consideration, the value function $\theta$ defined in Equation (\ref{eq:value_func_theta}) can be solved with classical Euler scheme based on grids. 
As shown in the left subfigure of Figure \ref{fig:multi_value}, the shape of the slices $\theta(\cdot, r)$ is very stable through time. 
We then apply the above quadratic approximation method and obtain its asymptotic form $\hat{\theta}(r)$ as $T\rightarrow+\infty$, which is shown to 
be very close to $\theta(0,r)$ in Figure \ref{fig:multi_value}. 
\begin{figure}[!h]
  \centering
  \begin{subfigure}{.55\textwidth}
    \centering
    \includegraphics[width=\linewidth]{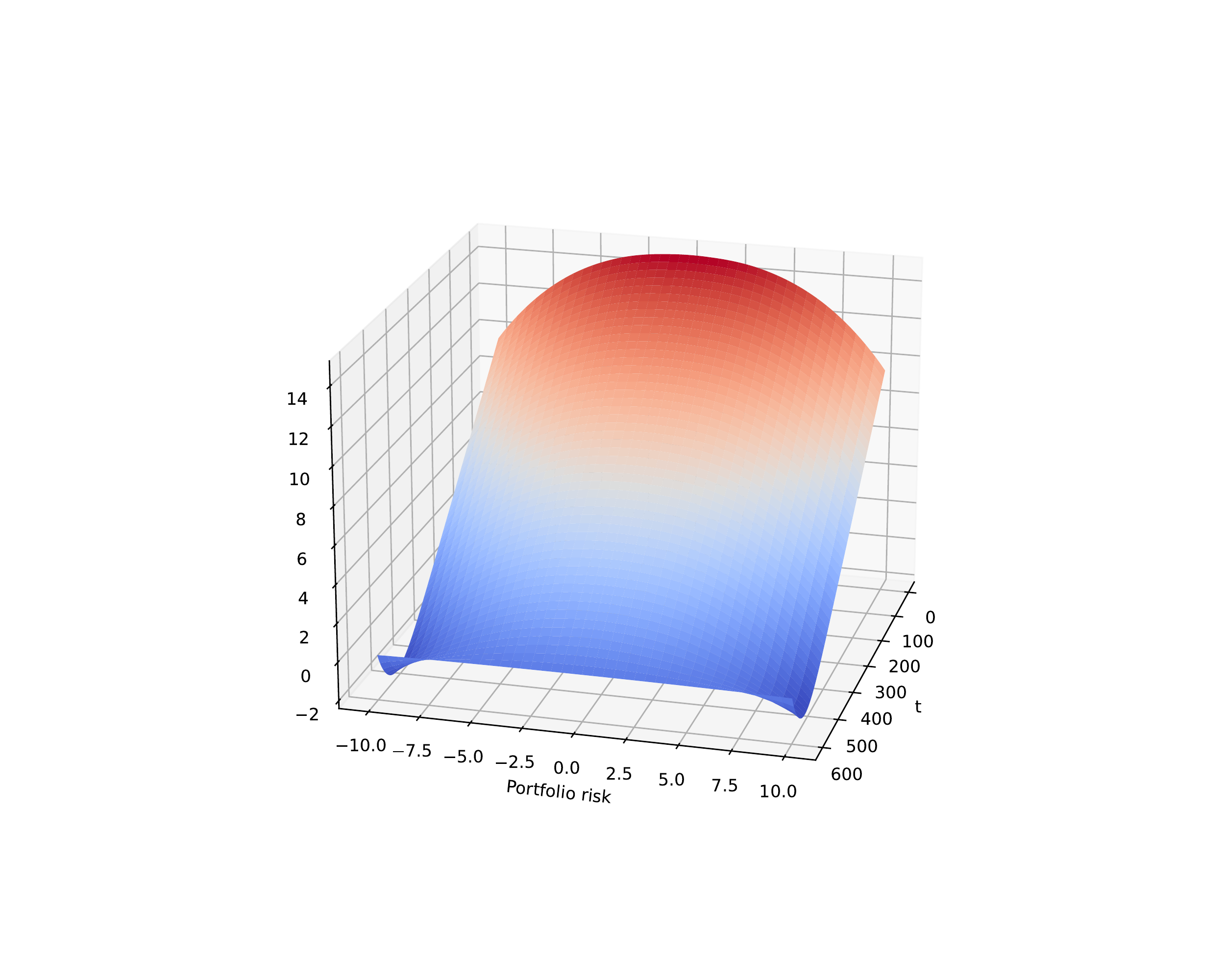}
  \end{subfigure}%
  \begin{subfigure}{.45\textwidth}
    \centering
    \includegraphics[width=\linewidth]{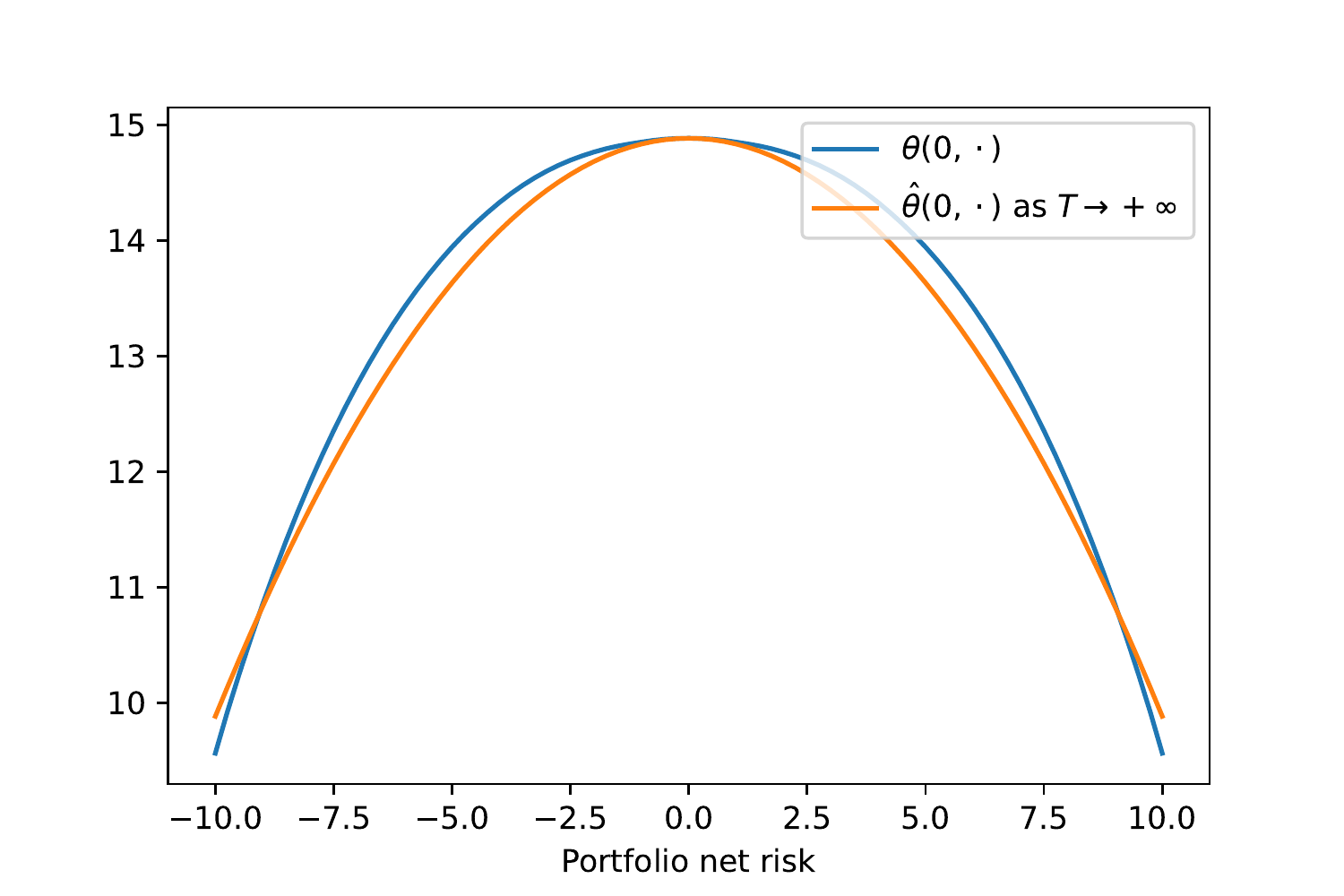}
  \end{subfigure}
  \caption{The value function $\theta(t,r)$ (left), and the asymptotic solution $\hat{\theta}(0,r)$ as $T\rightarrow T$ (right). 
        Note that the constant $C(0)$ in the expression of $\hat{\theta}(0, r)$ is chosen to have $\max_{r}\hat{\theta}(0,r)=\max_{r}\theta(0,r)$.}
  \label{fig:multi_value}
\end{figure}

\vskip 0.15in
\noindent
Similarly to the above example, we use backtests to perceive more directly the effectiveness of the approximated solution. 
We choose $T_b=\frac{T}{2}=300$ seconds and for each pair $(\kappa, \beta^{l})$ and $(\kappa, \beta^{\tilde{l}})$, the mean and standard deviation of the 
performance on $N=5000$ paths are computed. Figure \ref{fig:multi_backtest} shows the results with $\kappa$ varied from $10^{-4}$ to $1$.
We remark again that the mean-risk profile of $\beta^{\tilde{l}}$ is very close to that of $\beta^{l}$.   
\begin{figure}
  \centering
  \includegraphics[width=.8\textwidth]{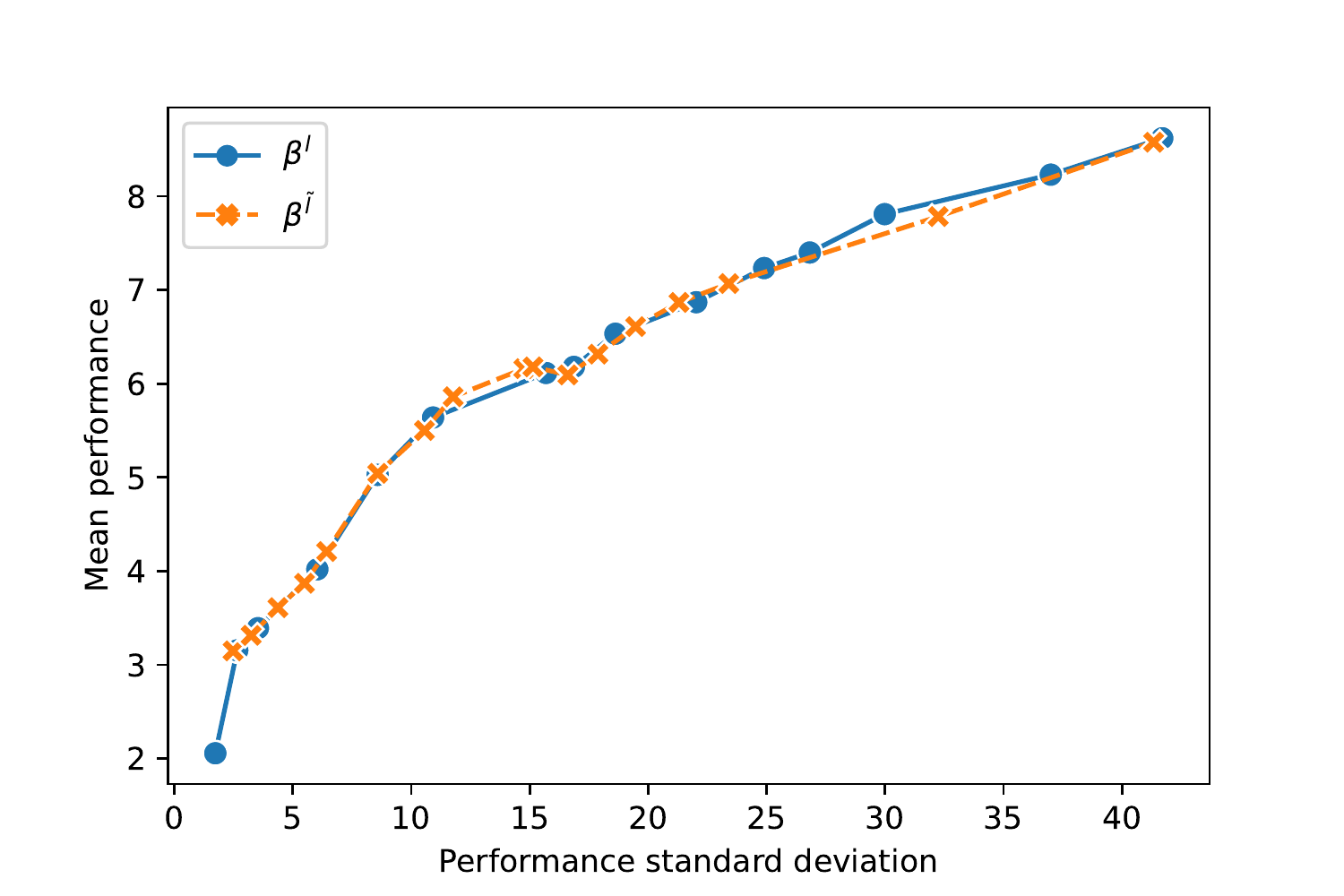}
  \caption{Mean performance and risk of the market making strategies $\beta^{l}$ and $\beta^{\tilde{l}}$ with varying risk penalization parameter $\kappa$ on four assets.}
  \label{fig:multi_backtest}
\end{figure}

\subsection{Example 3}
In this part we consider a market making problem on the six assets mentioned above to exemplify the effect of risk mutualization among them. We test with $T_b=2000$ seconds, $\kappa$ varying from $10^{-4}$ to $10^{-1}$ and $\kappa^j=\frac{\kappa}{2}, j=1,\cdots,6$. As indicated in the introduction, having a closed-form asymptotic solution allows us to recalibrate instantly the solution with the most recently observed parameters. To illustrate the idea, during the test we update $(\delta^j)_{j=1,\dots,d}$ and $\delta$ every 100 seconds according to the running state of the QRH model, \textit{i.e.} $X:=(S,Z^1,\cdots,Z^{10})$. As for the benchmarking case, we set $\kappa=0$ so that the inventory risk of each asset is managed individually. Classical Euler scheme is used for solving the value function of each asset, and then the corresponding optimal decisions can be obtained. Figure \ref{fig:online_backtest} gives the results on 5000 simulated paths. Clearly, the performance of uni-asset market making is suboptimal compared to the other setting, which also considers the net risk at the portfolio level.  

\begin{figure}
  \centering
  \includegraphics[width=.8\textwidth]{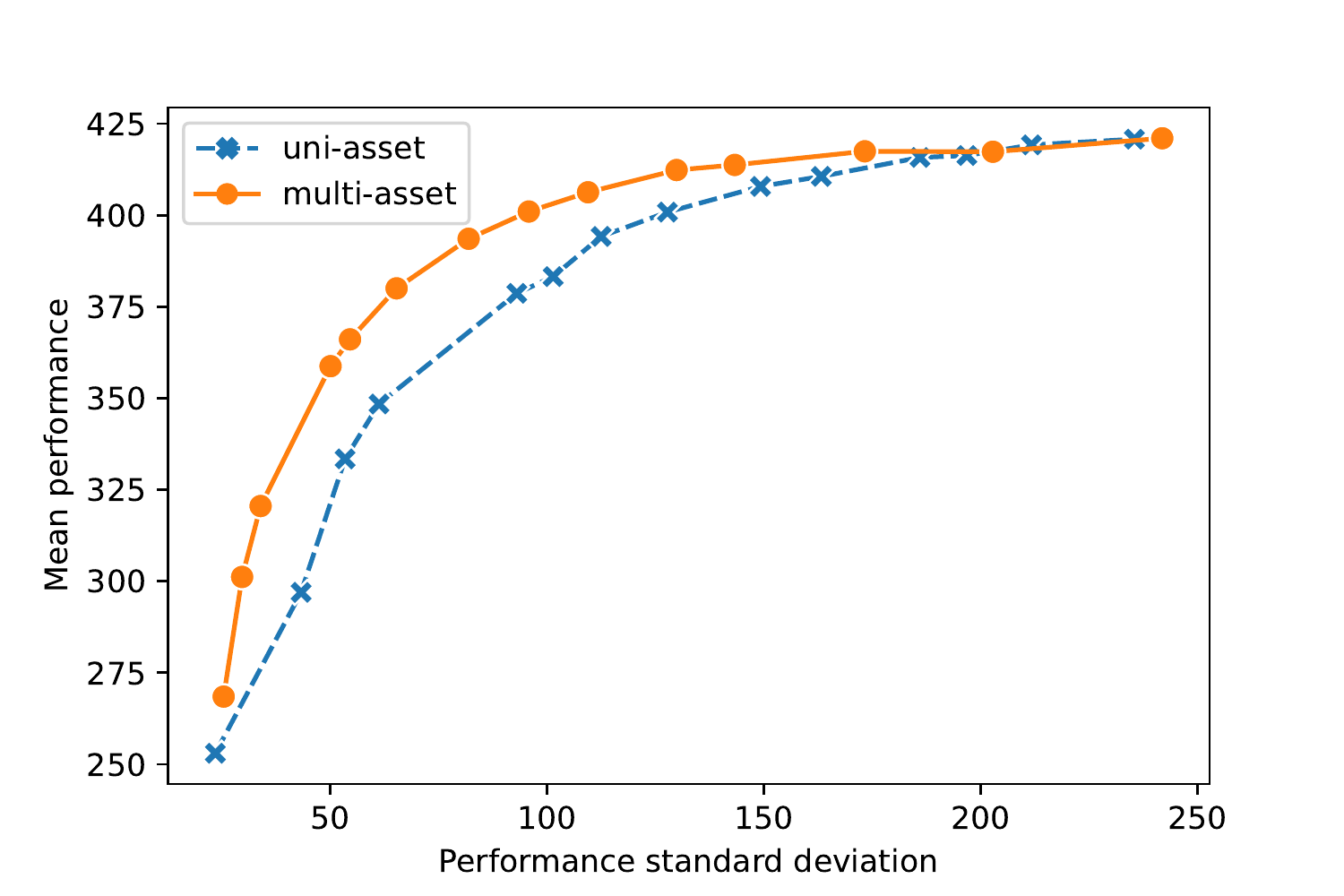}
  \caption{Mean performance and risk of the uni-asset and multi-asset market making strategies with varying risk penalization parameter $\kappa$. For the uni-asset case, optimal market making is conducted on
  each individual asset and the final profits are summed.}
  \label{fig:online_backtest}
\end{figure}

\clearpage
\bibliography{ref}
\bibliographystyle{abbrv}

\end{document}